PCMethodDiagram
\documentclass[twoside,american]{elsarticle}
\usepackage[T1]{fontenc}
\usepackage[latin9]{inputenc}
\usepackage{rotfloat}
\usepackage{mathrsfs}
\usepackage{amstext}
\usepackage{amsthm}
\usepackage{graphicx}

\makeatletter

\providecommand{\tabularnewline}{\\}

\theoremstyle{plain}
\newtheorem{thm}{\protect\theoremname}
\theoremstyle{definition}
\newtheorem{defn}[thm]{\protect\definitionname}
\ifx\proof\undefined
\newenvironment{proof}[1][\protect\proofname]{\par
	\normalfont\topsep6\p@\@plus6\p@\relax
	\trivlist
	\itemindent\parindent
	\item[\hskip\labelsep\scshape #1]\ignorespaces
}{%
	\endtrivlist\@endpefalse
}
\providecommand{\proofname}{Proof}
\fi
\theoremstyle{remark}
\newtheorem{rem}[thm]{\protect\remarkname}

\usepackage{amsmath,amssymb,amsfonts}

\usepackage{epstopdf}

\theoremstyle{plain}

\theoremstyle{definition}

\theoremstyle{remark}

\def\utilde#1{\mathord{\vtop{\ialign{##\crcr
$\hfil\displaystyle{#1}\hfil$\crcr\noalign{\kern1.5pt\nointerlineskip}
$\hfil\tilde{}\hfil$\crcr\noalign{\kern1.5pt}}}}}

\usepackage{hyperref}

\usepackage{colortbl}
\definecolor{gray1}{rgb}{0.9, 0.9, 0.9}
\definecolor{gray2}{rgb}{0.95, 0.95, 0.95}

\@ifundefined{showcaptionsetup}{}{%
 \PassOptionsToPackage{caption=false}{subfig}}
\usepackage{subfig}
\makeatother

\usepackage{babel}
\providecommand{\definitionname}{Definition}
\providecommand{\remarkname}{Remark}
\providecommand{\theoremname}{Theorem}

\begin{document}

\begin{frontmatter}{}

\title{Inconsistency indices for incomplete pairwise comparisons matrices}

\author{Konrad Ku\l akowski\corref{cor1}}

\ead{kkulak@agh.edu.pl}

\address{AGH University of Science and Technology, Faculty of Electrical Engineering,
Automatics, Computer Science and Biomedical Engineering, Kraków, Poland}

\author{Dawid Talaga\corref{cor2}}

\ead{talagadawid@gmail.com}

\address{The Higher Theological Seminary of the Missionaries, Kraków, Poland}
\begin{abstract}
Comparing alternatives in pairs is a very well known technique of
ranking creation. The answer to how reliable and trustworthy ranking
is depends on the inconsistency of the data from which it was created.
There are many indices used for determining the level of inconsistency
among compared alternatives. Unfortunately, most of them assume that
the set of comparisons is complete, i.e. every single alternative
is compared to each other. This is not true and the ranking must sometimes
be made based on incomplete data.

In order to fill this gap, this work aims to adapt the selected twelve
existing inconsistency indices for the purpose of analyzing incomplete
data sets. The modified indices are subjected to Monte Carlo experiments.
Those of them that achieved the best results in the experiments carried
out are recommended for use in practice.
\end{abstract}

\end{frontmatter}{}

\section{Introduction}

\subsection{On pairwise comparisons}

People have been making decisions since time began. Some of them are
very simple and come easily but other, more complicated, ones require
deeper analysis. Often, when many complex objects are compared, it
is difficult to choose the best one. The pairwise comparisons (PC)
method may help to solve this problem. Probably the first well-documented
case of using the PC method is the voting procedure proposed by \emph{Ramon
Llull }\citep{Colomer2011rlfa} - a thirteenth century alchemist and
mathematician. In \emph{Llull's} algorithm, the candidates were compared
in pairs - one against the other, and the winner was the one who won
in the largest number of direct comparisons. Later on, in the eighteenth
century, \emph{Llull's} voting system was reinvented by \emph{Condorcet}
\citep{Kulakowski2016note}. The next step came from \emph{Fechner}
\citep{Fechner1860edp} and \emph{Thurstone} \citep{Thurstone1927tmop}
who enabled the method to be used quantitatively for assessing intangible
social quantities. In the twentieth century, the PC method was a significant
component of the social choice and welfare theory \citep{Arrow1950adit,Sen1977scta}.
Currently, the PC method is very often associated with \emph{The Analytic
Hierarchy Process (AHP)}. In his seminal work on \emph{AHP,} \emph{Saaty}
\citep{Saaty1977asmf} combined a hierarchy together with pairwise
comparisons, which allowed the comparison of significantly more complex
objects than was possible before. In this work, we deal with the quantitative
and multiplicative PC method, that is, the basis of AHP.

The \emph{PC} method stems from the observation that it is much easier
for a man to compare objects pair by pair than to assess all the objects
at once . However, comparing in pairs presents us with various challenges.
One of them is the selection of the priority deriving method, including
the case when the set of comparisons is incomplete. Another, equally
important, one is the situation in which different comparisons may
lead to different or, even worse, opposing conclusions. All these
questions are extensively debated in the literature \citep{Ho2017tsot,Jablonsky2015aosp,Kulakowski2015otpo,Saaty1998rbev}.
However, one of them does not seem to have been sufficiently explored
- the co-existence of inconsistency and incompleteness. Namely, one
of the assumptions of \emph{AHP} says that the higher the inconsistency
of the set of paired comparisons, the lower the reliability of the
ranking computed. This assumption has its supporters \citep{Saaty1977asmf}
and opponents \citep{Forman1993fafa}, however, in general, most researchers
agree that high inconsistency may be the basis for challenging the
results of the ranking. The concept of inconsistency in the PC method
has been thoroughly studied and resulted in a number of works \citep{Brunelli2016atno,Brunelli2015apoi,Brunelli2013iifp,Koczkodaj2013oaoi}.
The original PC method assumes that each alternative has to be compared
with each other. However, researchers and practitioners quickly realized
that making so many comparisons can be difficult and sometimes even
impossible. For this reason, they proposed methods for calculating
the ranking based on incomplete sets of pairwise comparisons \citep{Lundy2016tmeo,Pan2014arpb,Srdjevic2014fltb,Fedrizzi2013osii,Fedrizzi2007ipca,Harker1987amoq}.

\subsection{Motivation}

Inconsistency of complete pairwise comparisons is well understood
and thoroughly studied in the literature \citep{Brunelli2018aaoi}.
One can easily find at least a dozen well-known indices allowing to
determine the level of inconsistency. In the case of incomplete PC
matrices, however, the phenomenon of inconsistency remains a relatively
little explored area. There are only a few proposals of inconsistency
indices for incomplete PC. One of them has been proposed by \emph{Harker}
\citep{Harker1987amoq}, later on, developed by \emph{Wedley} \citep{Wedley1993cpfi}.
The more recent index comes from \emph{Oliva et al}. \citep{Oliva2017sada}.
\emph{Bozóki} et al. proposed using the value of the logarithmic least
square criterion as the inconsistency measure \citep{Bozoki2010ooco}. 

The purpose of this work is to provide the readers with other inconsistency
indices for incomplete PC. However, the authors decided not to create
new indices but to adapt existing ones so that they could be used
in the context of incomplete PC. As a result, new versions for eight
inconsistency indices have been proposed including \emph{Koczkodaj's}
index \citep{Koczkodaj1993ando}, triad based indices \citep{Kulakowski2014tntb},
\emph{Salo and Hämäläinen} index \citep{Salo1995ppta}, geometric
consistency index \citep{Crawford1985taos,Aguaron2003tgci}, \emph{Golden-Wang}
index \citep{Golden1989aamo}, \emph{Barzilai's} relative error \citep{Barzilai1998cmfp}. 

One can expect that a useful inconsistency index should be resistant
to random deletion of comparisons (a random increase of incompleteness).
Thus, during the Montecarlo experiment, all the newly redefined indices,
including Harker's index and Bozóki's criterion, were compared for
their robustness in a situation of random data deletion. The proposed
approach allows assessing the credibility of the considered indices
concerning incompleteness. 

In AHP, the assessment of ranking veracity is inseparably linked to
the concept of inconsistency indices. The purpose of our paper is
to propose a variety of inconsistency indexes for incomplete PC and
identify those that may be particularly useful. We also realize that
the relationship between inconsistency and incompleteness must be
subject to further study \citep{Kulakowski2019tqoi}. Despite the
preliminary nature of Montecarlo results, we believe that this work
will contribute to the increase in the popularity of incomplete pairwise
comparisons as the ranking method and make it more reliable and trustworthy.

\subsection{Article organization}

The fundamentals of the pairwise comparisons method, including priority
deriving algorithms for complete and incomplete paired comparisons
and the concept of inconsistency, are introduced in Section \ref{sec:Preliminaries}.
Due to the relatively large number of indices considered in this work,
they are briefly reviewed in Section \ref{sec:Inconsistency-indexes}.
In (Section \ref{sec:Inconsistency-indexes-for}), we briefly describe
the main assumptions of our proposal for the extensions of selected
indexes. In particular, we propose dividing the indices into two groups:
the matrix based indices and the ranking based indices. The extensions
of the indices from the first group are described in (Section \ref{sec:Matrix-based-indices}),
while modifications of indices from the second group can be found
in (Section \ref{sec:Ranking-based-indices}). Proposals for extensions
considered in (Sections \ref{sec:Inconsistency-indexes-for} - \ref{sec:Ranking-based-indices})
are followed by a numerical experiment carried out in order to assess
the impact of incompleteness on the disturbances of the considered
indices (Section \ref{sec:Numerical-experiment}). Discussion and
summary (Section \ref{sec:Discussion}) close the article.

\section{Preliminaries\label{sec:Preliminaries}}

\subsection{The Pairwise Comparisons Method}

The \emph{PC} method is used to create a ranking of alternatives.
Let us denote them by $A=\{a_{1},\ldots,a_{n}\}$. Creating a ranking
in this case means assigning to each alternative a certain positive
real number $w(a_{i})$, called weight or priority. To achieve this,
the pairwise comparisons method compares each alternative with all
the others, then, based on all these comparisons, computes the priorities
for all alternatives. As alternatives are compared by experts in pairs,
it is convenient to represent the set of paired comparisons in the
form of a pairwise comparisons (PC) matrix.
\begin{defn}
\label{def:pc-matrix}The matrix $C$ is said to be a PC matrix 
\[
C=\left(\begin{array}{cccc}
1 & c_{12} & \cdots & c_{1n}\\
\vdots & 1 & \cdots & \vdots\\
\vdots & \cdots & \ddots & \vdots\\
c_{n1} & \cdots & c_{n,n-1} & 1
\end{array}\right),
\]
if $c_{ij}\in\mathbb{R}_{+}$ corresponds to the direct comparisons
of the i-th and j-th alternatives.
\end{defn}
For example: if, in an expert's opinion, the i-th alternative is two
times more preferred than the j-th alternative, then $c_{ij}$ receives
the value $2$. Of course, in such a situation it is natural to expect
that the j-th alternative is two times less preferred than the i-th
alternative, which, in turns, leads to $c_{ji}=1/2$.
\begin{defn}
The PC matrix in which $c_{ij}=1/c_{ji}$ is said to be \emph{reciprocal},
and this property is called \emph{reciprocity}.
\end{defn}
In further considerations in this paper, we will deal only with reciprocal
matrices. If the expert is indifferent when comparing $a_{i}$ and
$a_{j}$ then the corresponding pairwise comparisons result in $c_{ij}=1$,
which means a tie between the compared alternatives.

In practice, it is very often assumed that the results of pairwise
comparisons fall into a certain real and positive interval $1/s\leq c_{ij}\leq s$.
The value $s$ determines the range of the scale. For example, Saaty
\citep{Saaty2005taha} recommends the use of a discrete scale where
$c_{ij}\in\{1/9,1/8,\ldots1/2,1,2,\ldots,8,9\}$. Other researchers,
however, suggest other scales \citep{Franek2014jsac,Yuen2014cls,Dong2008acso}.
For the purpose of this article, we assume that $1/s\leq c_{ij}\leq s$
where $s=9$. 

Based on the PC matrix, the priorities of individual alternatives
are calculated (Fig. \ref{fig:pc-method}).

\begin{figure}[h]
\begin{centering}
\includegraphics[scale=0.5]{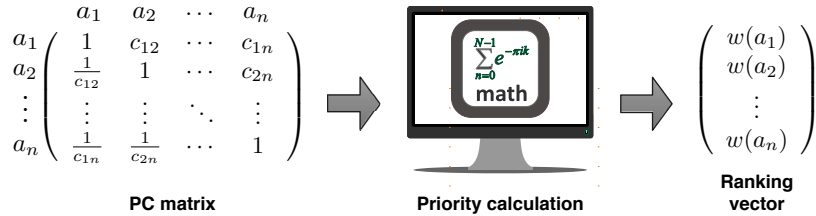}
\par\end{centering}
\caption{The PC method}
\label{fig:pc-method}
\end{figure}
 It is convenient to present them in the form of a weight vector (\ref{eq:weight-vector})
so that the i-th position in the vector denotes the weight of the
i-th alternative. 
\begin{equation}
w=\left[w(a_{1}),w(a_{2}),\ldots,w(a_{n})\right]^{T}.\label{eq:weight-vector}
\end{equation}

There are many procedures enabling the construction of a priority
vector. The first, and probably still the most popular, is one using
the eigenvector of $C$. According to this approach, referred to in
the literature as EVM (The Eigenvalue Method), the principal eigenvector
of the PC matrix is adopted as the priority vector $w$. For convenience,
the principal eigenvector is rescaled so that all its entries add
up to $1$. Formally, let 
\begin{equation}
Cw_{\textit{max}}=\lambda_{\textit{max}}w_{\textit{max}}\label{eq:evm-method-1}
\end{equation}
be the matrix equation so that $\lambda_{\textit{max}}$ is a principal
eigenvalue (spectral radius) of $C$. Then $w_{\textit{max}}$ is
a principal eigenvector of $C$ (due to the Perron-Frobenus theorem,
such a real and positive one exists \citep{Meyer2000maaa}). Thus,
the ranking vector $w_{\textit{ev}}$ is given as (\ref{eq:weight-vector})
where
\begin{equation}
w_{\textit{ev}}(a_{i})=\frac{w_{\textit{max}}(a_{i})}{\sum_{i=1}^{n}w_{\textit{max}}(a_{i})}.\label{eq:evm-method-2}
\end{equation}
Another popular priority deriving procedure is called GMM (geometric
mean method) \citep{Crawford1987tgmp}. In this approach, the priority
of an individual alternative is defined as an appropriately rescaled
geometric mean of the i-th row of $C$. Thus, the priority of the
i-th alternative is formally given as: 
\begin{equation}
w_{\textit{gm}}(a_{i})=\frac{\sqrt[n]{\prod_{j=1}^{n}c_{ij}}}{\sum_{i=1}^{n}w_{\textit{gm}}(a_{i})}.\label{eq:gmm-method}
\end{equation}
There are many other priority deriving methods \citep{Jablonsky2015aosp,Kou2014acmm,Lundy2016tmeo}.
In general, all of them lead to the same ranking vector unless the
set of paired comparisons is inconsistent. Let us look at inconsistency
a little bit closer.

\subsection{Inconsistency\label{subsec:Inconsistency}}

If we compare two pairs of alternatives $(a_{i},a_{k})$ and $(a_{k},a_{j})$,
then the results of these two comparisons also provide us with information
about the mutual relationship between $a_{i}$ and $a_{j}$. Indeed,
the result of the comparisons $a_{i}$ vs. $a_{k}$ is a positive
and real number $c_{ik}$ being an approximation of the ratio between
the priorities of the i-th and k-th alternatives i.e. 
\begin{equation}
\frac{w(a_{i})}{w(a_{k})}\approx c_{ik}.\label{eq:inc-eq-1}
\end{equation}
 Similarly, 
\begin{equation}
\frac{w(a_{k})}{w(a_{j})}\approx c_{kj}.\label{eq:inc-eq-2}
\end{equation}
This implies, of course, that
\begin{equation}
c_{ik}c_{kj}\approx c_{ij}.\label{eq:inc-eq-3}
\end{equation}
If a PC matrix is consistent, then the above formula turns into equality,
i.e. $c_{ik}c_{kj}=c_{ij}$ for every $i,k$, and $j\in\{1,\ldots,n\}$
where $i\neq k,\,k\neq j$ and $i\neq j$. Let us define these three
values $c_{ik}c_{kj}$ and $c_{ij}$ formally.
\begin{defn}
\label{def:triad-def}A group of three entries $\left(c_{ik},c_{kj},c_{ij}\right)$
of the PC matrix $C$ is called a triad if $i,j,k\in\{1,\ldots,n\}$
and $i\neq j,j\neq k$ and $i\neq k$.
\end{defn}
The above considerations also allow the introduction of the definition
of inconsistency.
\begin{defn}
\label{def:A-PC-matrix}A PC matrix $C=[c_{ij}]$ is said to be inconsistent
if there is a triad $c_{ik},c_{kj}$ and $c_{ij}$ for $i,j,k\in\{1,\ldots,n\}$
such that $c_{ik}c_{kj}\neq c_{ij}$. Otherwise $C$ is consistent.
\end{defn}
For the purpose of the article, any three values in the form $c_{ik},c_{kj}$
and $c_{ij}$ will be called a \emph{triad}. If there is a \emph{triad}
such that $c_{ik}c_{kj}\neq c_{ij}$ then the triad and, as follows,
the matrix $C$ are said to be \emph{inconsistent}.

It is easy to observe that when the PC matrix is consistent then (\ref{eq:inc-eq-1})
and (\ref{eq:inc-eq-2}) are also equalities. Hence, the consistent
PC matrix takes the form: 
\[
C=\left(\begin{array}{cccc}
1 & \frac{w(a_{1})}{w(a_{2})} & \cdots & \frac{w(a_{1})}{w(a_{n})}\\
\vdots & 1 & \cdots & \vdots\\
\vdots & \cdots & \ddots & \vdots\\
\frac{w(a_{n})}{w(a_{1})} & \cdots & \frac{w(a_{n})}{w(a_{n-1})} & 1
\end{array}\right).
\]

In practice, the PC matrix arises during tedious and error-prone work
of the experts who compare alternatives pair by pair. Therefore, due
to various reasons, inconsistency occurs. Since the data that we use
to calculate rankings are inconsistent, the question arises concerning
the extent to which the obtained ranking is credible. A highly inconsistent
PC matrix can mean that the expert preparing the matrix was inattentive,
distracted or just lacking sufficient knowledge and skill to carry
out the assessment. Therefore, most researchers agree that highly
inconsistent PC matrices result in unreliable rankings and should
not be considered. On the other hand, if the PC matrix is not too
inconsistent, the ranking can be successfully calculated. To determine
what the inconsistency level of the given PC matrix is, inconsistency
indices are used. Because there are over a dozen of them (sixteen
indices are subjected to the Montecarlo experiment described in this
work), their exact description has been included in Section \ref{sec:Inconsistency-indexes}

\subsection{Incompleteness\label{subsec:Incompleteness}}

As stated above, the PC matrix contains mutual comparisons of all
alternatives taken into account. However, from a practical point of
view, completing all necessary comparisons can be difficult. The first
reason is the square increase in the number of comparisons in relation
to the number of alternatives considered (providing reciprocity $n$
alternatives implies at least $n(n-1)/2$ comparisons). As it is easy
to see, for $7$ alternatives we need $21$ comparisons but for $9$
we need as many as $36$ comparisons and so on. Therefore, in the
case of a large number of alternatives, gathering all comparisons
is just labor-intensive. This is especially true as these comparisons
are usually made by experts who, as always, suffer from a lack of
time. For that reason, Wind and Saaty \citep{Wind1980maot} indicated
the optimal number of alternatives as $7\pm2$. Another reason for
the lack of comparison can be the inability of the expert to compare
two alternatives. The source of this impossibility may be ethical
or moral doubts or a weaker knowledge of the particular issue \citep{Harker1987ipci}.
All the above reasons led to the necessity of introducing incomplete
PC matrices, that is, ones in which some entries are not defined.
For the purpose of the article, the missing (undefined) comparisons
in the PC matrix are denoted as $?$. Let us define incomplete PC
matrices formally.
\begin{defn}
\label{def:an-incomplete-pc-matrix-def}A PC matrix $C=[c_{ij}]$
is said to be an incomplete PC matrix if $c_{ij}\in\mathbb{R}_{+}\cup\{?\}$
where $c_{ij}=?$ means that the comparison of the i-th and j-th alternatives
is missing.
\end{defn}
In the case of missing values, the reciprocity condition would mean
that $c_{ij}=?$ implies $c_{ji}=?$.

Bearing in mind all the above problems with obtaining a complete set
of comparisons, Harker \citep{Harker1987amoq,Harker1987ipci} proposed
the extension of EVM for an incomplete PC matrix.  HM (The Harker's
method) requires the creation of an auxiliary matrix $B=[b_{ij}]$,
in which

\begin{equation}
b_{ij}=\begin{cases}
c_{ij} & \text{if}\,c_{ij}\,\text{is a real number greater than}\,0\\
0 & \text{otherwise}\\
m_{i} & \text{is the number of unanswered questions in the i-th row of}\,C
\end{cases}.\label{eq:harker-matrix-form}
\end{equation}
Finding and scaling a principal eigenvector of $B$ leads directly
to the desired numerical ranking.

The well-known GMM (\ref{eq:gmm-method}) also has its own extension
for the incomplete PC matrices \citep{Bozoki2010ooco}. According
to the ILLS (incomplete logarithmic least square) method, one needs
to solve the linear equation: 
\begin{eqnarray}
R\widehat{w} & = & g\label{eq:bfr-log-equation}\\
\widehat{w}(a_{1}) & = & 0
\end{eqnarray}

where $R=[r_{ij}]$ is the Laplacian matrix \citep{Merris1994lmog}
such that
\begin{equation}
r=\begin{cases}
\alpha & \text{if\,}i=j\,\text{where}\,\alpha\,\text{is the number of ? in the i-th row}\\
-1 & \text{if}\,c_{ij}\neq?\\
0 & \text{if}\,c_{ij}=?
\end{cases},\label{eq:r-matrix-bfr-method}
\end{equation}

$g$ is the constant term vector $g=\left[g_{1},\ldots,g_{n}\right]^{T}$
where
\[
g_{i}=\log\prod_{\substack{c_{ij}\neq?\\
j=1,\ldots,n
}
}c_{ij},
\]

and $\widehat{w}$ is the logarithmized priority vector $\widehat{w}=\left[\widehat{w}(a_{i}),\ldots,\widehat{w}(a_{n})\right]$,
i.e. $\widehat{w}(a_{i})=\log w(a_{i})$ for $i=1,\ldots,n$ where
$w$ is the appropriate priority vector\footnote{In practice, $w$ should also be rescaled so that all its entries
sum up to $1$}. The ILLS approach is also optimal in the same sense as the GMM method
is \citep{Crawford1987tgmp,Bozoki2010ooco}. It is worth to note that
the above method can also be formulated in terms of the geometric
mean \citep{kulakowski2019otgm}.

In addition to these two methods, there are also other ways to deal
with incomplete matrices, for example, the entropy approach \citep{Pan2014arpb}
or the spanning tree approach \citep{Siraj2012east}.

\subsection{Graph representation}

It is often convenient to consider a set of pairwise comparisons,
a PC matrix, as a graph. For this reason, let us introduce the definition
of a graph of the given PC matrix.
\begin{defn}
\label{def:graph-of-a-matrix}A directed graph $T_{C}=(V,E,L)$ is
said to be a graph of $C$ if $V=\{a_{1},\ldots,a_{n}\}$ is a set
of vertices, $E\subset V^{2}\backslash\bigcup_{i=1}^{n}(a_{i},a_{i})$
is a set of ordered pairs called directed edges, $L:V^{2}\rightarrow\mathbb{R}_{+}$
such that $L(a_{i,},a_{j})=c_{ij}$ is the labeling function, and
$C=[c_{ij}]$ is the $n\times n$ PC matrix.
\end{defn}
For example, let us consider the following incomplete PC matrix in
which $c_{34}$ and $c_{43}$ are undefined: 
\[
C=\left(\begin{array}{cccc}
1 & \frac{2}{3} & \frac{4}{3} & \frac{1}{2}\\
\frac{3}{2} & 1 & 2 & \frac{3}{4}\\
\frac{3}{4} & \frac{1}{2} & 1 & ?\\
2 & \frac{4}{3} & ? & 1
\end{array}\right)
\]
The graph $T_{C}$ is shown in Fig. \ref{fig:assoc-graph}. Providing
that the matrix $C$ is reciprocal\footnote{In the literature, non-reciprocal PC matrices are also considered
\citep{Kulakowski2016srot,Hovanov2008dwfg}.}, the upper triangle of $C$ contains all the information necessary
to create a ranking. Therefore, instead of the graph $T_{C}$ one
can analyze a graph of the upper triangle of $C$. Thanks to this,
we obtain a simplified drawing of the graph, without losing essential
information. Let $\text{UT}(C)$ denote the upper triangle of $C$.
The graph of the upper triangle of $C$ is shown in Fig. \ref{fig:assoc-graph-up-triangle}.

\begin{figure}[h]
\begin{centering}
\subfloat[$T_{C}$ - graph of $C$]{\centering{}\includegraphics[scale=0.5]{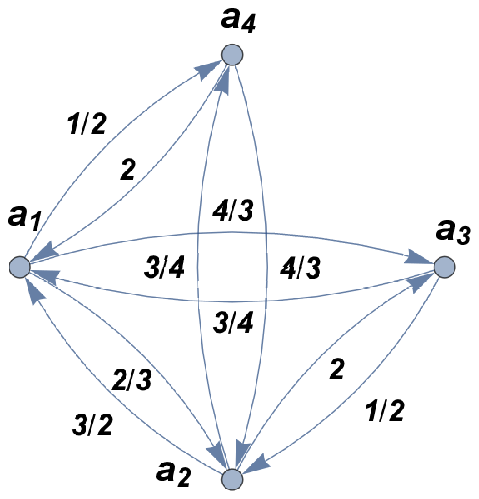}\label{fig:assoc-graph}}~~~\subfloat[$T_{\text{UP}(C)}$ - the graph of the upper triangle of $C$]{\centering{}\includegraphics[scale=0.5]{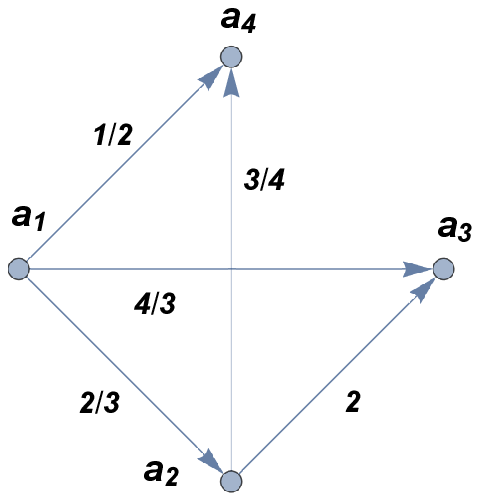}\label{fig:assoc-graph-up-triangle}}
\par\end{centering}
\caption{Graph representations of the matrix $C$}
\end{figure}

For incomplete PC matrices, one property of the graph is particularly
important. This is strong connectivity \citep{Quarteroni2000nm}.
\begin{defn}
A directed graph $T_{C}$ is strongly connected if for any pair of
distinct vertices $a_{i}$ and $a_{j}$ there is an oriented path
$(a_{i},a_{r_{1}})$, $(a_{r_{1}},a_{r_{2}})$, $\ldots,(a_{r_{k}},a_{j})$
in $E$ from $a_{i}$ to $a_{j}$.
\end{defn}
The matrix $C$ for which $T_{C}$ is strongly connected is called
\emph{irreducible} \citep{Quarteroni2000nm}. It is easy to notice
that two distinct alternatives $a_{i}$ and $a_{j}$ (through the
oriented path leading from $a_{i}$ to $a_{j}$) can be compared together
only if $T_{C}$ is strongly connected. This immediately leads to
the observation that only when the PC matrix is irreducible (i.e.
the appropriate graph is strongly connected), are we able to compute
the ranking \citep{Siraj2012east,Harker1987amoq}. For this reason,
in the article, we only deal with irreducible PC matrices.

For the purposes of this article, let us also define the concept of
the cycle in a graph.
\begin{defn}
\label{def:cycle-def-1}An ordered sequence of distinct vertices $p=a_{i_{1}},a_{i_{2}},\ldots,a_{i_{m}}$
such that $\left\{ a_{i_{1}},a_{i_{2}},\ldots,a_{i_{m}}\right\} \in V$
is said to be a path between $a_{i_{1}}$ and $a_{i_{m}}$ with the
length $m-1$ in $T_{C}=(V,E,L)$ if $(a_{i_{1}},a_{i_{2}})$, $(a_{i_{2}},a_{i_{3}}),\ldots$,$(a_{i_{m-1}},a_{i_{m}})\in E$.
\end{defn}
and similarly
\begin{defn}
\label{def:cycle-def}A path $s$ between $a_{i_{1}}$ and $a_{i_{m}}$
with the length $m-1$ is said to be a simple cycle with the length
$m$ if also $(a_{i_{m}},a_{i_{1}})\in E$.
\end{defn}
In the case of a cycle, it is not important which element in the sequence
of vertices is first. Thus, if $p=a_{i_{1}},a_{i_{2}},\ldots,a_{i_{m}}$
is a cycle then $q=a_{i_{2}},\ldots,a_{i_{m}},a_{i_{1}}$ means the
same cycle, i.e. $p=q$. 
\begin{defn}
\label{def:all-cycles}Let $T_{C}=(V,E,L)$ be a graph of $C$. Then
the set of all paths between $a_{i}$ and $a_{j}$ in $T_{C}$ is
defined as $\mathcal{P}_{C,i,j}\overset{\textit{df}}{=}\{p=a_{i_{1}},a_{i_{2}},\ldots,a_{i_{m}}\,\,$
is a path between $\text{\,}\,a_{i_{1}}\text{and}\,\,a_{i_{m}}\,\text{in}\,\,T_{C}\}$.
Similarly, the set of all cycles longer than $q$ in $T_{C}$ is defined
as $\mathcal{S}_{C,q}\overset{\textit{df}}{=}\{s=a_{i_{1}},a_{i_{2}},\ldots,a_{i_{m}}\,\,$
is a cycle of $C\,\,\text{for}\,m>q\}$.
\end{defn}
The number of adjacent edges to the given vertex $a$ usually is called
the degree of $a$ and written as $\deg(a)$. Based on the degree
of vertex we can define a degree matrix. 
\begin{defn}
\label{def:degree-matrix}Let $T_{\textit{UP}(C)}=(V,E,L)$ be a graph
of $C$. The degree matrix $D=[d_{ij}]$ of $T_{\textit{UP}(C)}$
is a diagonal matrix such that $d_{ii}=\deg(a_{i})$ for and $d_{ij}=0$
for $i,j=1,\ldots,n$ and $i\neq j$.
\end{defn}

\section{Inconsistency indices \label{sec:Inconsistency-indexes}}

In his seminal work, \emph{Saaty} \citep{Saaty1977asmf} proposed
a measurement of inconsistency as a way of determining credibility
of the ranking. Since then, many inconsistency indices have been created
allowing the degree of inconsistency in the set of paired comparisons
to be determined \footnote{As in the paper we deal with cardinal (quantitative) pairwise comparisons,
we do not consider ordinal inconsistency of the ordinal pairwise comparisons.
A good example of the ordinal inconsistency index is the generalized
consistency coefficient \citep{Kulakowski2018iito}.}. Below, we briefly present several inconsistency indices being the
subject of extension as well as a few already existing inconsistency
indexes for incomplete matrices. A systematic review of the various
inconsistency indexes can be found in Brunelli \citep{Brunelli2018aaoi}. 

\subsection{Inconsistency indices for complete PC matrices}

The list of indices is opened by  the geometric consistency index
(GCI). GCI given as: 
\begin{equation}
I_{\textit{G}}=\frac{2}{(n-1)(n-2)}\sum_{i=1}^{n}\sum_{j=i+1}^{n}\log^{2}e_{ij}.\label{eq:GCI}
\end{equation}
where 
\begin{equation}
e_{ij}=c_{ij}\frac{w(a_{j})}{w(a_{i})},\,\,\,\,\,\,i,j=1,...,n,\label{eq:gci-error}
\end{equation}
was proposed by \emph{Crawford} and \emph{Williams} \citep{Crawford1985anot},
and then called as the geometric consistency index by Aguaròn and
Moreno-Jimènez \citep{Aguaron2003tgci}. 

In contrast to the previous indices, a measure defined by \emph{Koczkodaj}
does not examine the average inconsistency of the set of paired comparisons
\citep{Koczkodaj1993ando,Duszak1994goan}. Instead, it spots the highest
local inconsistency and adopts it as an inconsistency of the examined
matrix. A local inconsistency is determined by means of the triad
index $K_{i,k,j}$ defined as follows:

\begin{equation}
K_{i,j,k}=\min\left\{ \left|1-\frac{c_{ik}c_{kj}}{c_{ij}}\right|,\left|1-\frac{c_{ij}}{c_{ik}c_{kj}}\right|\right\} .\label{eq:koczkodaj-triad-index}
\end{equation}
The inconsistency index for $C$ obtains the form:
\begin{equation}
K=\max\left\{ K_{i,j,k}\,\,|\,\,1\leq i<j<k\leq n\right\} \label{eq:K}
\end{equation}

\emph{Ku\l akowski} and \emph{Szybowski} proposed two other inconsistency
indices \citep{Kulakowski2014tntb}, which are also based on triads\footnote{The first of them was later proposed by Grzybowski \citep{Grzybowski2016nroi}}.
They both use the Koczkodaj triad index $K_{ijk}$ (\ref{eq:koczkodaj-triad-index}).
The indices are designed as the average of all possible $K_{ijk}$
given as follows\footnote{Note that $\binom{n}{3}=\frac{n(n-1)(n-2)}{6}$.}:

\begin{equation}
I_{1}=\frac{6\sum_{\{i,j,k\}\in T}K_{ijk}}{n(n-1)(n-2)},\label{eq:I1}
\end{equation}

\begin{equation}
I_{2}=\frac{6\sqrt{\sum_{\{i,j,k\}\in T}K_{ijk}^{2}}}{n(n-1)(n-2)},\label{eq:I2}
\end{equation}
where $T=\left\{ \{i,j,k\}:i\neq j,i\neq k,j\neq k\,\text{and}\,1\leq i,j,k\leq n\right\} $.
Both indices can be combined together to create new coefficients.
Based on this observation,\emph{ }the authors proposed two parametrized
families of indices: 

\begin{equation}
I_{\alpha}=\alpha K+(1-\alpha)I_{1},\label{eq:Ia}
\end{equation}
where $0\leq\alpha\leq1$, and 

\begin{equation}
I_{\alpha,\beta}=\alpha K+\beta I_{1}+(1-\alpha-\beta)I_{2},\label{eq:Iab}
\end{equation}
where $0\leq\alpha+\beta\leq1$. 

\emph{Golden} and \emph{Wang} proposed another inconsistency index
\citep{Golden1989aamo}. According to this approach the priority vector
was calculated using the geometric mean method, then scaled to add
up to 1. In this way, the vector $g^{*}=[g_{1,}^{*},...,g_{n}^{*}]$
was obtained, where $C=[c_{ij}]$ is an $n$ by $n$ PC matrix. Then,
every column is scaled so that the sum of its elements is $1$. Let
us denote the matrix with the rescaled columns by $C^{*}=[c_{ij}^{*}]$.
The inconsistency index is defined as follows:
\begin{equation}
\textit{GW}=\frac{1}{n}\sum_{i=1}^{n}\sum_{j=1}^{n}\mid c_{ij}^{*}-g_{i}^{*}\mid.\label{eq:gw-index}
\end{equation}

The index proposed by \emph{Salo} and \emph{Hämäläinen} \citep{Salo1995ppta}
requires an auxiliary interval matrix $R$ to be prepared. In this
matrix, every element is a pair corresponding to the highest and the
lowest approximation of $c_{ij}$. The $n\times n$ matrix $R$ is
given as

\begin{equation}
R=\left(\begin{array}{ccc}
\left(\underline{r}_{11},\overline{r}_{11}\right) & \ldots & \left(\underline{r}_{1n},\overline{r}_{1n}\right)\\
\vdots & \ddots & \vdots\\
\left(\underline{r}_{n1},\overline{r}_{n1}\right) & \ldots & \left(\underline{r}_{nn},\overline{r}_{nn}\right)
\end{array}\right),
\end{equation}
where 
\[
\underline{r_{ij}}=\min\left\{ c_{ik}c_{kj}\mid k=1,\ldots,n\right\} ,\,\,\,\text{and}\,\,\,\overline{r_{ij}}=\max\left\{ c_{ik}c_{kj}\mid k=1,\ldots,n\right\} 
\]
As every $c_{ik}c_{ji}$ is an approximation of $c_{ij}$ then $\underline{r_{ij}}$
is the lowest and $\overline{r_{ij}}$ is the highest approximation
of $c_{ij}$. Finally, the inconsistency index is:
\begin{equation}
\textit{\ensuremath{I_{\textit{SH}}}}=\frac{2}{n(n-1)}\sum_{i=1}^{n-1}\sum_{j=i+1}^{n}\frac{\overline{r}_{ij}-\underline{r}_{ij}}{\left(1+\overline{r}_{ij}\right)\left(1+\underline{r}_{ij}\right)}.
\end{equation}
For further reference, see \citep{Brunelli2015itta}. 

The last of the extended indexes was proposed by\emph{ Barzilai} \citep{Barzilai1998cmfp}.
It requires calculation of the weight vector using the arithmetic
mean method for each row and the preparation of two auxiliary matrices.
Let us denote $\Delta_{i}=\frac{1}{n}\sum_{j=1}^{n}\widehat{c}_{ij},$
where $\widehat{C}=[\widehat{c}_{ij}]$ is an $n$ by $n$ additive
PC matrix, i.e. such that $c_{ij}\in\mathbb{R}$ and $c_{ij}=-c_{ji}$.
The two auxiliary matrices are given as follows: $X=\left[x_{ij}\right]=\left[\Delta_{i}-\Delta_{j}\right]$,
$E=\left[e_{ij}\right]=\left[\widehat{c}_{ij}-x_{ij}\right]$. Ultimately,
the formula for the \emph{relative error} (considered as the inconsistency
index) is as follows:
\begin{equation}
\textit{RE}(\widehat{C})=\frac{\sum_{ij}e_{ij}^{2}}{\sum_{ij}\widehat{c}_{ij}^{2}}.
\end{equation}
Of course, $\textit{RE}$ was defined for additive PC matrices. Thus,
for the purpose of multiplicative PC matrices, \emph{Barzilai} proposes
to transform it using a $\log$ function with any base. Thus, for
the PC matrix $C=[c_{ij}]$ and $\widehat{C}=[\log c_{ij}]$ we obtain:
$\textit{RE}(C)\overset{\textit{df}}{=}\textit{RE}(\widehat{C})$.

\subsection{Inconsistency indices for incomplete PC matrices}

Inconsistency indexes for incomplete PC matrices are definitely less
than for complete matrices. The first of them, probably the earliest
defined is the Saaty's index for incomplete PC matrices defined by
Harker \citep{Harker1987amoq} as: 
\[
\widetilde{\textit{CI}}=\frac{\widetilde{\;\lambda}_{\textit{max}}-n}{n-1}
\]
where $\widetilde{\lambda}_{\textit{max}}$ is the principal eigenvalue
of the auxiliary matrix $B$ (\ref{eq:harker-matrix-form}).  Following
\citep[p. 356]{Harker1987amoq}, the consistency index $\widetilde{\textit{CI}}$
can also be written as:

\[
\widetilde{\textit{CI}}=-1+\frac{1}{n(n-1)}\left(\sum_{i=1}^{n}m_{i}+\sum_{\begin{array}{c}
1\leq i<j\leq n\\
c_{ij}\neq?
\end{array}}\left(c_{ij}\frac{\widetilde{w}(a_{j})}{\widetilde{w}(a_{i})}+c_{ji}\frac{\widetilde{w}(a_{i})}{\widetilde{w}(a_{j})}\right)\right).
\]
where $\widetilde{w}=\left[\widetilde{w}(a_{1}),\ldots,\widetilde{w}(a_{n})\right]^{T}$
is the principal eigenvector of $B$.  Properties of $\widetilde{\textit{CI}}$
were also tested in \citep{Talaga2018ioip}.

Another method of measuring the inconsistencies of incomplete matrices
was proposed by Bozóki et al. \citep{Bozoki2010ooco}. According to
this approach all the possible completion of an incomplete PC matrix
are considered, then one that minimizes a certain inconsistency criterion
is chosen. The value of this criterion for the selected completion
can be considered as the inconsistency value of the given incomplete
PC matrix. Adopting the function: 
\[
\textit{LLS}(C,w)=\sum_{\substack{i,j=1\\
i\neq j
}
}^{n}\left(\log c_{ij}-\frac{\log w(a_{i})}{\log w(a_{j})}\right)^{2}
\]
as such criterion \citep{Crawford1987tgmp,Bozoki2010ooco} leads to
ILLS method (Sec. \ref{subsec:Incompleteness}).  However, $\textit{LLS}(C,w)$
can also be treated as a ranking based inconsistency index. Thus,
following \citep{Bozoki2010ooco}, we can adopt 
\[
\widetilde{L}(C)\overset{\textit{df}}{=}\textit{LLS}(Q_{C},w_{\textit{ILLS}})
\]
 as the inconsistency index for incomplete matrix $C$, where $Q_{C}=[q_{ij}]$
is an optimal completion of $C$ defined as: 
\[
q_{ij}=\begin{cases}
c_{ij} & \text{if}\,c_{ij}\neq?\\
\frac{w_{\textit{ILLS}}(a_{i})}{w_{\textit{ILLS}}(a_{j})} & \text{if}\,c_{ij}=?
\end{cases}.
\]
The last index considered has been proposed by Oliva et al. \citep{Oliva2017sada}
as 
\[
O(\mathcal{R})=\rho(D^{-1}\mathcal{S})-1.
\]
In the above equation $\mathcal{R}$ is an incomplete multiplicative
PC matrix in which every missing element is represented by $0$, matrix
$D$ is the degree matrix of the graph $T_{\textit{UP}(\mathcal{R})}$
(Def. \ref{def:degree-matrix}), $\rho$ stands for spectral radius
of a matrix, and $\mathcal{S}=\mathcal{R}-\textit{Id}$, where $\textit{Id}$
denotes $n\times n$ identity matrix.

\section{Extensions of inconsistency indexes for incomplete matrices\label{sec:Inconsistency-indexes-for}}

Among the indices listed above, two distinct groups can be distinguished.
The first group consists of indices based on the concept of a triad
(Def. \ref{def:triad-def}) and the idea of the triad's inconsistency
(Section \ref{subsec:Inconsistency}). According to this idea, the
analysis of three different entries of a PC matrix is able to reveal
the inconsistency. Of course, this analysis is local as it is limited
to three specific comparisons. However, if we take into account all
possible triads in $C$, our judgment as to the inconsistency will
become global and may act as an inconsistency index. In this approach,
the ranking method is not important. Inconsistency is estimated directly
using elements of the PC matrix and the definition of inconsistency
(Def. \ref{def:A-PC-matrix}). We will call all the indices for which
the above observation holds \emph{the matrix based indices}. This
group includes:
\begin{itemize}
\item Koczkodaj's inconsistency index,
\item Triad based average inconsistency indices,
\item Salo and Hamalainen index.
\end{itemize}
The second group of indices are those for which calculation of the
ranking is indispensable. The general idea behind all of the indices
in this group is that the ratio $w(a_{i})/w(a_{j})$ needs to be similar
or even identical to the value $c_{ij}$ for all $i,j\in\{1,\ldots,n\}$
(see \ref{eq:inc-eq-1} - \ref{eq:inc-eq-3}). Of course, to verify
the difference between $w(a_{i})/w(a_{j})$ and $c_{ij}$ we first
have to compute the ranking vector (\ref{eq:weight-vector}). For
this reason, each of the indices in this group is closely related
to some priority deriving method. For the purpose of this article,
we will call them the ranking based indices. This group includes\footnote{For the purpose of the Montecarlo experiment we also consider Harker's
extension of Saaty's consistency index \citep{Harker1987amoq}, Logarithmic
least square criterion \citep{Bozoki2010ooco} and Oliva et al. inconsistency
index \citep{Oliva2017sada}.}:
\begin{itemize}
\item Geometric consistency index,
\item Golden-Wang index,
\item Relative Error index
\end{itemize}
Of course, the above division is, to some extent, arbitrary. For example,
the geometric consistency index can be expressed using triad based
local inconsistency \citep{Brunelli2013anot}. The Harker's extension
of Saaty's consistency index can also be estimated using Koczkodaj's
consistency index \citep{Kulakowski2015otpo}. Despite the existing
relationships between different inconsistency indices, there is no
global framework unifying them. Indeed, considering the work on the
axiomatization of inconsistency indexes \citep{Brunelli2015apoi,Koczkodaj2014oaoi,Brunelli2016raoi,Koczkodaj2018aoii}
it can be assumed that the creation of such a framework will be one
of the challenges for researchers shortly.

\subsection{Matrix based indices\label{subsec:Matrix-based-indices}}

Matrix based indices use triads to determine inconsistency. However,
in an incomplete PC matrix some triads might be missing. For example,
if $c_{ik}$ is undefined the triad $c_{ik},c_{kj}$ and $c_{ij}$
is also undefined. Of course, one may think that analysis of the remaining
triads allows us to assess the degree of matrix inconsistency. Unfortunately,
it is not true and in some circumstances this strategy fails. This
happens when the PC matrix does not have any triads and yet it is
irreducible. Let us consider the following PC matrix:

\begin{equation}
C=\left(\begin{array}{ccccccc}
1 & 1/2 & ? & ? & ? & ? & 1/7\\
2 & 1 & ? & 6 & 4 & 2 & ?\\
? & ? & 1 & 4 & 3 & 3/2 & ?\\
? & 1/6 & 1/4 & 1 & ? & ? & 1/2\\
? & 1/4 & 1/3 & ? & 1 & ? & 1/4\\
? & 1/2 & 2/3 & ? & ? & 1 & 1/3\\
7 & ? & ? & 2 & 4 & 3 & 1
\end{array}\right)\label{eq:example_matrix}
\end{equation}

The graph $T_{\text{UP}(C)}$ of the upper triangle of the above matrix
is shown in Fig. \ref{fig:no-triad-graph-ex}. As we can see, there
is no cycle\footnote{Remember that in the full graph $T_{C}$ each edge $(a_{i},a_{j})$
has its counterpart $(a_{j},a_{i})$.} (Def. \ref{def:cycle-def}) with the length $3$, thus there are
no triads in $C$. It is easy to observe that $T_{C}$ is strongly
connected, thus $C$ is irreducible, and therefore it is a valid incomplete
PC matrix.

\begin{figure}[h]
\begin{centering}
\includegraphics[scale=0.5]{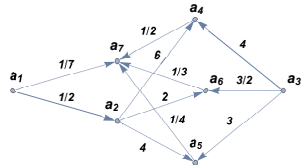}
\par\end{centering}
\caption{$T_{\text{UP}(C)}$ - the graph of the upper triangle of the PC matrix
$C$ which does not contain triads}
\label{fig:no-triad-graph-ex}
\end{figure}

Since we can not use triads in this case, the question arises as to
whether we should not use quadruples $(c_{ik},c_{kr},c_{rj},c_{ij})$
i.e. cycles with the length $4$. As in the previous case, the answer
is negative. We are able to construct a directed graph without cycles
with the length $4$ and, as follows, a PC matrix which does not contain
quadruples. Fortunately, there are some cycles in most strongly connected
graphs of the PC matrices. The only exceptions are the strongly connected
graphs with $n$ vertices and only $2(n-1)$ edges. In such graphs,
every vertex $a_{i}$ is connected with only one other vertex $a_{j}$
by two edges $(a_{i},a_{j})$ and $(a_{j},a_{i})$. If we remove any
pair of the existing edges $\{(a_{i},a_{j}),(a_{j},a_{i})\}$, the
graph would cease to be strongly connected. Conversely, if we add
a new pair of edges $\{(a_{p},a_{q}),(a_{q},a_{p})\}$ to $E$ it
would be a cycle in the graph. For this reason, whenever there are
cycles in $T_{C}$, we will try to use them all to inconsistency without
limiting their length or quantity. However, if there are no cycles
in the graph, the concept of inconsistency loses meaning. Therefore,
the only option is to accept that the considered PC matrix is consistent
(an alternative would be to assume that the inconsistency is indeterminate).

In order to use the cycle to determine the matrix inconsistency, let
us extend the Definition \ref{def:A-PC-matrix}.
\begin{defn}
\label{def:A-PC-matrix-extended}A PC matrix $C=[c_{ij}]$ is said
to be inconsistent if there exists a cycle $a_{i_{1}},\ldots,a_{i_{m}}$
in $T_{C}$ such that $c_{i_{1}i_{2}}c_{i_{2}i_{3}}\cdot\ldots\cdot c_{i_{m-1}i_{m}}\neq c_{i_{1}i_{m}}$.
Otherwise $C$ is consistent.
\end{defn}
For complete PC matrices, both definitions \ref{def:A-PC-matrix}
and \ref{def:A-PC-matrix-extended} are equivalent. To prove that,
it is enough to show that wherever $C$ is inconsistent in the sense
(Def. \ref{def:A-PC-matrix}) then it is also inconsistent in the
sense (Def. \ref{def:A-PC-matrix-extended}), and reversely inconsistency
in the sense (Def. \ref{def:A-PC-matrix-extended}) entails inconsistency
in the sense (Def. \ref{def:A-PC-matrix}).
\begin{thm}
\label{thm:Inc-equiv-theorem}Every complete PC matrix $C$ is inconsistent
in the sense of (Def. \ref{def:A-PC-matrix}) if and only if it is
inconsistent in the sense of (Def. \ref{def:A-PC-matrix-extended})
\end{thm}
\begin{proof}
``$\Rightarrow$'' Let $C$ be inconsistent in the sense of (Def.
\ref{def:A-PC-matrix}) i.e. there is a triad such that $c_{ik}c_{kj}\neq c_{ij}$.
Since this triad is also a cycle with the length $3$, $C$ is also
inconsistent in the sense of (Def. \ref{def:A-PC-matrix-extended}).

``$\Leftarrow$'' Let $C$ be inconsistent in the sense of (Def.
\ref{def:A-PC-matrix-extended}) i.e. there is a cycle $s=a_{i_{1}},\ldots,a_{i_{m}}$
such that $c_{i_{1}i_{2}}c_{i_{2}i_{3}}\cdot\ldots\cdot c_{i_{m-1}i_{m}}\neq c_{i_{1}i_{m}}$,

and let us suppose for a moment that $C$ is consistent in the sense
of (Def. \ref{def:A-PC-matrix}). The latter assumption means that
every triad is consistent, thus, in particular it also holds that
$c_{i_{1}i_{2}}c_{i_{2}i_{3}}=c_{i_{1}i_{3}}$. Therefore, the first
assumption can be written in the form: $c_{i_{1}i_{3}}\cdot\ldots\cdot c_{i_{m-1}i_{m}}\neq c_{i_{1}i_{m}}$.
Applying the same reasoning many times, we subsequently get that $c_{i_{1}i_{4}}\cdot\ldots\cdot c_{i_{m-1}i_{m}}\neq c_{i_{1}i_{m}}$,
$c_{i_{1}i_{5}}\cdot\ldots\cdot c_{i_{m-1}i_{m}}\neq c_{i_{1}i_{m}}$
and finally $c_{i_{1}i_{m-1}}\cdot c_{i_{m-1}i_{m}}\neq c_{i_{1}i_{m}}$.
However, as we assume that every triad is consistent, therefore also
$a_{i_{1}},a_{i_{m-1}},a_{m}$ is consistent, thus it must hold that
$c_{i_{1}i_{m-1}}\cdot c_{i_{m-1}i_{m}}=c_{i_{1}i_{m}}$. Contradiction.
\end{proof}
Of course, when $C$ is incomplete the two above definitions of inconsistency
are not equivalent. In particular, there may be PC matrices which
do not have any triads, hence they have to be considered as consistent,
but the same matrices may have graphs with cycles longer than $4$
that might be inconsistent. In general, however, the cycle based definition
of inconsistency is more general than the Def. \ref{def:A-PC-matrix}.
We may observe the following property.
\begin{rem}
Every PC matrix (complete and incomplete) inconsistent in the sense
(Def. \ref{def:A-PC-matrix}) is also inconsistent in the sense of
(Def. \ref{def:A-PC-matrix-extended}), but not reversely.

Definition \ref{def:A-PC-matrix-extended} also allows us to quantify
the inconsistency. As we will see later on, the ratio:

\begin{equation}
R_{s}\overset{\textit{df}}{=}\frac{c_{i_{1}i_{2}}\cdot\ldots\cdot c_{i_{m-1}i_{m}}}{c_{i_{1}i_{m}}}\label{eq:cycle-ratio}
\end{equation}

defined for a cycle $s=a_{i_{1}},\ldots,a_{i_{m}}$ is a useful way\footnote{Note that when $C$ is reciprocal then $R_{s}$ does not depend on
the choice of $m$. Indeed: 
\[
R_{s}=\frac{c_{i_{1}i_{2}}c_{i_{2}i_{3}}\cdot\ldots\cdot c_{i_{m-1}i_{m}}}{c_{i_{1}i_{m}}}=\frac{\frac{1}{c_{i_{2}i_{1}}}c_{i_{2}i_{3}}\cdot\ldots\cdot c_{i_{m-1}i_{m}}}{\frac{1}{c_{i_{m}i_{1}}}}=\frac{c_{i_{2}i_{3}}\cdot\ldots\cdot c_{i_{m-1}i_{m}}c_{i_{m}i_{1}}}{c_{i_{2}i_{1}}}=\ldots
\]
} for measuring inconsistency within the set of $m$ alternatives $a_{i_{1}},\ldots,a_{i_{m}}$.
We use this fact to define several matrix based indices for an incomplete
PC matrix. The idea of using cycles for inconsistency measurement
can be found in \citep{Bozoki2019tlls,Koczkodaj2013oaoi}.
\end{rem}

\subsection{Ranking based indices\label{subsec:Ranking-based-indices}}

The ranking based indices need the results of ranking in order to
calculate inconsistency. The considered indices use two different
priority deriving methods: EVM (\ref{eq:evm-method-1}) and GMM (\ref{eq:evm-method-2}).
Although both methods have been defined for a complete PC matrix,
they have their counterparts for incomplete PC matrices (Section \ref{subsec:Incompleteness}).
These are the HM \citep{Harker1987ipci} and ILLS approaches \citep{Bozoki2010ooco}.

The starting point of both extensions is the assumption that every
missing value $c_{ij}$ in $C$ should eventually take the value $w(a_{i})/w(a_{j})$.
Therefore, the authors of extensions replaced the unknown values by
the appropriate ratios, and then tried to solve such modified problems.
Let $\widehat{C}$ be a PC matrix obtained from $C$ by replacing
every $c_{ij}=?$ by $w(a_{i})/w(a_{j})$. In HM, the eigenvalue equation
(\ref{eq:evm-method-1}) takes the form: 
\[
\widehat{C}w=\lambda_{\textit{max}}w,
\]

and after the appropriate transformations, we finally get 
\[
Bw=\lambda_{\textit{max}}w
\]

where $B$ is an auxiliary matrix (\ref{eq:harker-matrix-form}).
Similarly, in the ILLS method \citep{Bozoki2010ooco}, the authors
define the priority of the i-th alternative as the geometric mean
of the i-th row of $\widehat{C}$. The adoption of this assumption
leads to the matrix equation (\ref{eq:bfr-log-equation}), whose solution
determines the desirable vector of priorities.

In general, the ranking based indices define the inconsistency as
the differences between $c_{ij}$ and $w(a_{i})/w(a_{j})$ for $i,j=1,\ldots,n$.
Since both HM and ILLS replace every $c_{ij}=?$ by the corresponding
ratio $w(a_{i})/w(a_{j})$, then the missing judgments do not contribute
to the inconsistency, but are considered as perfectly consistent.
From a practical point of view, during construction of the ranking
based indices for incomplete PC matrices, we can either ignore the
missing values or just assume that $c_{ij}=?$ equals $w(a_{i})/w(a_{j})$.

\section{Matrix based indices for incomplete PC matrices\label{sec:Matrix-based-indices}}

\subsection{Koczkodaj index}

The Koczkodaj index is directly based on the concept of a triad and
its inconsistency. Thus, as explained above (Sec. \ref{subsec:Matrix-based-indices}),
a triad's inconsistency has to be replaced by the cycle's inconsistency.
Let $C$ be an irreducible and incomplete PC matrix (Def. \ref{def:an-incomplete-pc-matrix-def})
and $T_{C}$ be a graph of $C$ (Def. \ref{def:graph-of-a-matrix}).
Then let us define the inconsistency of a single cycle\footnote{It is worth noting that if $s=a_{i}a_{k}a_{j}$ then $K_{s}=K_{i,k,j}$
(\ref{eq:koczkodaj-triad-index}, \ref{eq:koczkodaj-cycle-inconsistency}).} longer than\footnote{Cycles with the length $2$ are always consistent as $c_{ij}c_{ji}/c_{ii}=1$,
thus they are not relevant from the point of inconsistency of $C$.} $2$ i.e. $s\in\mathcal{S}_{C,2}$ as

\begin{equation}
K_{s}\overset{\textit{df}}{=}\min\left\{ \left|1-R_{s}\right|,\left|1-R_{s}^{-1}\right|\right\} \label{eq:koczkodaj-cycle-inconsistency}
\end{equation}
Then the Koczkodaj index for the incomplete PC matrix $C$ can be
defined as: 
\[
\widetilde{K}\overset{\textit{df}}{=}\begin{cases}
\max\left\{ K_{s}:s\in\mathcal{S}_{C,2}\right\}  & \left|\mathcal{S}_{C,2}\right|>0\\
0 & \left|\mathcal{S}_{C,2}\right|=0
\end{cases}
\]
The case in which $\left|\mathcal{S}_{C,2}\right|=0$ refers to the
situation when the $n\times n$ matrix $C$ is irreducible, but it
contains exactly $n-1$ comparisons i.e. $T_{\textit{UP}(C)}$ is
a tree \citep{Cormen2009ita}.

\subsection{Triad based average inconsistency indices}

The method of replacing triads with cycles can be successfully used
in the case of triad based average inconsistency indices. Thus, providing
that $C$ is an irreducible and incomplete PC matrix, we have: 
\[
\widetilde{I}_{1}\overset{\textit{df}}{=}\begin{cases}
\frac{\sum_{s\in\mathcal{S}_{C,2}}K_{s}}{\left|\mathcal{S}_{C,2}\right|} & \left|\mathcal{S}_{C,2}\right|>0\\
0 & \left|\mathcal{S}_{C,2}\right|=0
\end{cases},
\]
and correspondingly,

\[
\widetilde{I}_{2}\overset{\textit{df}}{=}\begin{cases}
\frac{\sqrt{\sum_{s\in\mathcal{S}_{C,2}}K_{s}^{2}}}{\left|\mathcal{S}_{C,2}\right|} & \left|\mathcal{S}_{C,2}\right|>0\\
0 & \left|\mathcal{S}_{C,2}\right|=0
\end{cases}.
\]
The $I_{\alpha}$ and $I_{\alpha,\beta}$ indices (\ref{eq:Ia}, \ref{eq:Iab})
also need to be changed accordingly.

\[
\widetilde{I}_{\alpha}\overset{\textit{df}}{=}\alpha\widetilde{K}+(1-\alpha)\widetilde{I}_{1},
\]

\[
\widetilde{I}_{\alpha,\beta}\overset{\textit{df}}{=}\alpha\widetilde{K}+\beta\widetilde{I}_{1}+(1-\alpha-\beta)\widetilde{I}_{2}.
\]

\subsection{Salo and Hamalainen index}

The $\textit{SHI}$ index is based on the observation that every product
$c_{ik}c_{kj}$ for any $k=1,\ldots,n$ is an approximation of $c_{ij}$
\citep{Salo1995ppta}.  When $C$ is irreducible, due to the strong
connectivity of $T_{C}$ between every two vertices $a_{q}$ and $a_{j}$
there is a path $p=a_{q},a_{i_{2}},a_{i_{2}},\ldots,a_{i_{m-1}},a_{j}$
(Def. \ref{def:cycle-def}). This means that $c_{q,i_{2}},c_{i_{2}i_{3},}\ldots,c_{i_{m-1},j}$
are defined. Let us denote the product induced by $p$ as $\pi_{p}=c_{q,i_{2}}c_{i_{2}i_{3},}\cdot\ldots\cdot c_{i_{m-1},j}$.
Due to (\ref{eq:inc-eq-1} - \ref{eq:inc-eq-2}), $\pi_{p}$ is also
a good approximation of $c_{ij}$. Thus, let us define 
\[
\underset{\widetilde{\,\,\,\,\,}}{r}{}_{ij}\overset{\textit{df}}{=}\min\left\{ \pi_{p}\mid p\in\mathcal{P}_{C,i,j}\right\} 
\]
and 
\[
\widetilde{r}_{ij}\overset{\textit{df}}{=}\max\left\{ \pi_{p}\mid p\in\mathcal{P}_{C,i,j}\right\} 
\]
The modified $I_{\textit{SH}}$ index can be defined as 
\[
\widetilde{I}_{\textit{SH}}\overset{\textit{df}}{=}\frac{2}{n(n-1)}\sum_{i=1}^{n-1}\sum_{j=i+1}^{n}\frac{\widetilde{r}_{ij}-\underset{\widetilde{\,\,\,\,\,}}{r}{}_{ij}}{\left(1+\widetilde{r}_{ij}\right)\left(1+\underset{\widetilde{\,\,\,\,\,}}{r}{}_{ij}\right)}.
\]

\section{Ranking based indices for incomplete PC matrices\label{sec:Ranking-based-indices}}

\subsection{Geometric consistency index}

The geometric consistency index (\ref{eq:GCI}) is directly based
on observation (\ref{eq:inc-eq-1}) i.e. $w(a_{i})/w(a_{k})\approx c_{ik}$,
where $w$ is the priority vector obtained from GMM \citep{Crawford1985taos,Aguaron2003tgci}.
In an incomplete PC matrix, a priority vector $\widetilde{w}$ has
to be computed using the ILLS method (Sec. \ref{subsec:Incompleteness}).
Following this method  wherever $c_{ij}=?$  it is replaced by $\widetilde{w}(a_{i})/\widetilde{w}(a_{j})$.
 This leads to the following formula: 
\begin{equation}
\widetilde{I}_{G1}\overset{\textit{df}}{=}\frac{2}{(n-1)(n-2)}\sum_{e\in\mathscr{E}}\log^{2}e,\label{eq:gci-1-expr}
\end{equation}
where $C=[c_{ij}]$, $\widetilde{w}=[\widetilde{w}(a_{1}),\ldots,\widetilde{w}(a_{n})]^{T}$
is the ranking vector calculated using the ILLS method, and $\mathscr{E}\overset{\textit{df}}{=}\left\{ e_{ij}=c_{ij}\frac{\widetilde{w}(a_{j})}{\widetilde{w}(a_{i})}\,:\,c_{ij}\neq?\,\,\text{and}\,\,i<j\right\} $.
The second way to extend the $\textit{GCI}$ index is to calculate
the average of all non-zero $\log^{2}e$ expressions also possible.
In this approach we assume that $c_{ij}=?$ does not contribute to
our knowledge about inconsistency.  Hence, the second version of
GCI for incomplete PC matrices is as follows:
\[
\widetilde{I}_{G2}\overset{\textit{df}}{=}\frac{1}{\left|\mathscr{E}\right|}\sum_{e\in\mathscr{E}}\log^{2}e.
\]

\subsection{Golden-Wang index}

The Golden-Wang index is based on the observation that every column
of a consistent PC matrix equals the ranking vector multiplied by
some constant scaling factor \citep{Golden1989aamo}.  Thus, after
scaling every column so that it sums up to $1$, it holds that $c_{ij}^{*}=w(a_{i})$,
where $C^{*}=[c_{ij}^{*}]$ is the consistent PC matrix with the rescaled
columns.  The difference between $c_{ij}^{*}$ and $w(a_{i})$ is
higher when the inconsistency is greater.

Despite the fact that the observations remain true for any priority
deriving method, the authors recommend using GMM (\ref{eq:gmm-method}).
The relationship between $c_{ij}^{*}$ and $w(a_{i})$  also remains
valid in the case of an incomplete PC matrix, however, due to the
missing values the scaling procedure needs to be modified. Let us
consider the k-th column of the irreducible incomplete PC matrix $C$
and the ranking vector $w$. Let every element of $C$ be either $1$
or $?$. So it is easy to see that the ranking vector is composed
of the same $1/n$ values. For example, for a $5\times5$ matrix we
may have: 
\[
\left[\begin{array}{c}
c_{1k}\\
c_{2k}\\
c_{3k}\\
c_{4k}\\
c_{5k}
\end{array}\right]=\left[\begin{array}{c}
1\\
1\\
1\\
?\\
?
\end{array}\right],\,\,\,\,\left[\begin{array}{c}
w(a_{1})\\
w(a_{2})\\
w(a_{3})\\
w(a_{4})\\
w(a_{5})
\end{array}\right]=\left[\begin{array}{c}
1/5\\
1/5\\
1/5\\
1/5\\
1/5
\end{array}\right],
\]
where $k=1,2$ or $3$. Then, after scaling (so that the sum of elements
is one) the k-th column is: 
\[
\left[\begin{array}{c}
c_{1k}\\
c_{2k}\\
c_{3k}\\
c_{4k}\\
c_{5k}
\end{array}\right]=\left[\begin{array}{c}
1/3\\
1/3\\
1/3\\
?\\
?
\end{array}\right]
\]
as the undefined elements cannot be scaled. It is evident that $c_{ik}^{*}\neq w(a_{i})$
for $i=1,2,3$ as $1/3\neq1/5$. The solution is to construct the
priority vector $w_{k}$ which has the missing values at the same
positions as the k-th column. Let us consider:
\[
\left[\begin{array}{c}
c_{1k}\\
c_{2k}\\
c_{3k}\\
c_{4k}\\
c_{5k}
\end{array}\right]=\left[\begin{array}{c}
1\\
1\\
1\\
?\\
?
\end{array}\right],\,\,\,\,\left[\begin{array}{c}
w_{k}(a_{1})\\
w_{k}(a_{2})\\
w_{k}(a_{3})\\
w_{k}(a_{4})\\
w_{k}(a_{5})
\end{array}\right]=\left[\begin{array}{c}
1/5\\
1/5\\
1/5\\
?\\
?
\end{array}\right]
\]
In such a case, after scaling, indeed $c_{ik}^{*}=w_{k}^{*}(a_{i})=1/3$
for $i=1,2,3$.

The above observation shows the way in which the Golden-Wang index
can be extended to incomplete pairwise comparisons. Let us define
this extension more formally. For this purpose, let us assume that
$C$ is an irreducible, incomplete PC matrix, and $w$ is the ranking
vector calculated using the ILLS method. Then let $\Omega=[\omega_{ij}]$
be an $n\times n$ matrix such that 
\[
\omega_{ij}\overset{\textit{df}}{=}\begin{cases}
w(a_{i}) & \text{if}\,\,c_{ij}\neq?\\
? & \text{if}\,\,c_{ij}=?
\end{cases}.
\]
Next, let us scale every column in $\Omega$ and in $C$ so that all
the elements in the column sum up to one (in both cases, undefined
values are omitted, i.e. only defined elements are subject to scaling).
As a result, we get two matrices $C^{*}=[c_{ij}^{*}]$ and $\Omega^{*}=[\omega_{ij}^{*}]$
with the appropriately rescaled columns. The absolute differences
between the entries of these two matrices form the Golden-Wang index
for incomplete PC matrices. Thus, following (\ref{eq:gw-index}),
we may define: 
\[
\widetilde{\textit{GW}}\overset{\textit{df}}{=}\frac{1}{n}\sum_{i=1}^{n}\sum_{j=1}^{n}\mid c_{ij}^{*}-\omega_{ij}^{*}\mid
\]

\subsection{Relative Error}

  In general, Barzilai's Relative Error $\textit{RE}$ has been
defined for additive PC matrices \citep{Barzilai1998cmfp}. For the
purpose of this paper we use its logarithmized version suitable for
multiplicative matrices \citep{Barzilai1998cmfp}. For an incomplete
PC matrix, similarly to the case of the $\textit{GCI}$ index, we
may assume that $w$ is calculated using the ILLS method (for multiplicative
matrices Barzilai's original approach uses GMM). The use of the ILLS
method implies the assumption that every $c_{ij}=?$ can be substituted
by $w_{\textit{ILLS}}(a_{i})/w_{\textit{ILLS}}(a_{j})$. In particular,
in such a case $e_{ij}=0$, hence in the formula (\ref{eq:relative-error-log-form})
they can be omitted. Thus, the relative error for the incomplete multiplicative
PC matrix $C$ takes the form:

\[
\widetilde{\textit{RE}}_{1}\overset{\textit{df}}{=}\frac{\sum_{c_{ij}\neq?}\left[\log c_{ij}\frac{w_{\textit{ILLS}}(a_{j})}{w_{\textit{ILLS}}(a_{i})}\right]^{2}}{\sum_{c_{ij}\neq?}\log^{2}c_{ij}+\sum_{c_{ij}=?}\log^{2}\frac{w_{\textit{ILLS}}(a_{i})}{w_{\textit{ILLS}}(a_{j})}},
\]
where $i,j=1,\ldots,n$. Another way of extending Relative Error to
incomplete, multiplicative PC matrices is to skip all the expressions
requiring missing comparisons. Similarly to the case of the geometric
consistency index, this leads to a shorter formula:

\[
\widetilde{\textit{RE}}_{2}\overset{\textit{df}}{=}\frac{\sum_{c_{ij}\neq?}\left[\log c_{ij}\frac{w_{\textit{ILLS}}(a_{j})}{w_{\textit{ILLS}}(a_{i})}\right]^{2}}{\sum_{c_{ij}\neq?}\log^{2}c_{ij}}.
\]

\section{Numerical experiment\label{sec:Numerical-experiment}}

Does increasing incompleteness affect inconsistency? Intuition suggests
that it should not. If decision makers responsible for creating PC
matrices are inconsistent in their judgments, then we may assume that
their inconsistency will not depend on whether they answer all or
only part of the questions. The only difference is that in the case
of a complete matrix the experts will more often do both: make mistakes
and respond correctly. Of course, we implicitly assume that the experts
are able to consider each question with similar attention, i.e. questions
are not too many, experts have enough time to think about them, and
they are professionals in the field. So if indeed inconsistency does
not depend on incompleteness, then the incomplete matrix can be treated
just as a sample of some complete PC matrix. Of course, there may
always be some differences between the sample and the entire population,
but it is natural to expect that they are reasonably small.

In light of these observations, it seems interesting to see how much
the inconsistency for the complete matrix will differ from its incomplete
sample with reference to the given inconsistency index, i.e. how robust
the given inconsistency index is for the PC matrix deterioration.
To this end, we created $1000$ consistent $7\times7$ PC matrices.
Then, the entries of each matrix were disturbed by multiplying them
by the randomly chosen coefficient $\gamma\in[1/d,d]$. We repeated
the disturbance procedure $30$ times for $\gamma=1,\ldots,30$. In
this way, we received the set $\mathcal{C}$ composed of $30000$
PC matrices with varying degrees of inconsistency.

Every complete $7\times7$ PC matrix contains $21$ comparisons (entries
above the diagonal). On the other hand, the smallest irreducible $7\times7$
PC matrix has $6$ comparisons (as at least six edges are needed to
connect seven different vertices of a graph of a matrix). Hence, preserving
irreducibility, at most $15$ comparisons can be safely removed from
the complete $7\times7$ PC matrix. Therefore, for each of the $30000$
complete PC matrices, we prepared 15 randomly incomplete irreducible
PC matrices, so that every complete matrix had its ``sample'' matrix
with $1,2$ up to $15$ missing comparisons. Finally, we received
$480000$ complete and incomplete PC matrices for which we calculated
all $13$ inconsistency indices defined in Sections \ref{sec:Matrix-based-indices}
and \ref{sec:Ranking-based-indices}.

In order to check the robustness of different inconsistency indices,
we calculate the directed distance between the inconsistency of the
complete matrix and their incomplete counterparts. Let $I(C)-I(C_{k})$
be the ordered distance between where $I(C)$ means the value of the
inconsistency index $I$ calculated for a complete matrix $C\in\mathcal{C}$,
and $I(C_{k})$ denotes inconsistency of the matrix that was obtained
from $C$ by removing $k$ comparisons determined by using $I$. Of
course, different indices may take values from various ranges. Therefore,
to allow those indices to be compared with each other, the ordered
distance is divided by the higher component of each difference. Hence,
the rescaled ordered distance $\Delta_{I}(C,C_{k})$ between the inconsistency
of two matrices $C$ and $C_{k}$ is given as: 
\[
\Delta_{I}(C,C_{k})\overset{\textit{df}}{=}\begin{cases}
\frac{I(C)-I(C_{k})}{\max\left\{ I(C),I(C_{k})\right\} } & \max\left\{ I(C),I(C_{k})\right\} >0\\
0 & I(C)=I(C_{k})=0
\end{cases}.
\]
The above formula also takes into account the situation where $I(C)=I(C_{k})=0$.
In such a case, it is assumed that $\Delta(C,C_{k})=0$. The final
result is the average ordered distance: 
\begin{equation}
D(I,k)\overset{\textit{df}}{=}\frac{1}{\left|\mathcal{C}\right|}\sum_{C\in\mathcal{C}}\Delta_{I}(C,C_{k}).\label{eq:average-ordered-distance}
\end{equation}
The subsequent values $D(I,0),D(I,1),\ldots,D(I,15)$ allow the difference
between the inconsistency of the complete and incomplete matrix to
be assessed with respect to the given index $I$ and the number of
missing comparisons $k$.

In an ideal case, $D(I,k)$ should be $0$ for all inconsistency indices
and every possible $k$. In practice, of course, it is impossible
as not all comparisons in the PC matrix are inconsistent to the same
extent. Therefore, it is possible that an incomplete matrix will be
less (or more) consistent than its complete counterpart. If the incomplete
matrix $C_{k}$ is less inconsistent than the complete matrix $C$,
i.e. $I(C)>I(C_{k})$, then $\Delta_{I}(C,C_{k})>0$. Reversely, if
the incomplete matrix is more inconsistent than its consistent predecessor,
the distance is negative i.e. $\Delta_{I}(C,C_{k})<0.$ In other words,
the sign (direction) of a distance $D_{I}$ informs us if there is
a greater complete or incomplete PC matrix. The closer $\Delta_{I}(C,C_{k})$
is to $0$, i.e. the smaller $\left|\Delta_{I}(C,C_{k})\right|$ is,
the more resistant to incompleteness is the index $I$. For the purpose
of this study, a directed distance for all inconsistency indices has
been computed using $30000$ matrices, then the results have been
averaged as $D(I,k)$. Any particular value of $D(I,k)$ for some
fixed $I$ and $k$ can be interpreted as an average value of directed
distance for a $7\times7$ randomly disturbed matrix where $I$ is
an inconsistency index, and $k$ is the number of missing elements.
Of course: the index is better (more robust) when it is closer to
the abscissa\footnote{When assessing the robustness of $I$ it is not important whether
$\Delta_{I}(C,C_{k})$ takes positive or negative values. How far
$\Delta_{I}(C,C_{k})$ is from the abscissa is more important, i.e.
the size of $\left|\Delta_{I}(C,C_{k})\right|$}. Of course it is possible that for some particular $Q,R,k$ the distance
$\left|D(Q,k)\right|$ is greater than $\left|D(R,k)\right|$. However,
for $k+1$ it may turn out that $\left|D(Q,k+1)\right|<\left|D(R,k+1)\right|$.
In such a case, it is difficult to indicate the winner, as one time
$Q$ is better, the other time $R$ is better. Therefore, as the final
measure of index robustness, we suggest taking the area between the
abscissa and the plot of its rescaled ordered distance (Fig. \ref{fig:directed-distance-ind-chart1}).
The discrete counterpart of the size of this area is the sum of the
absolute value of subsequent $D(I,k)$ i.e.

\[
\mathscr{D}(I)\overset{\textit{df}}{=}\sum_{k=0}^{15}\left|D(I,k)\right|.
\]
Of course, the smaller $\mathscr{D}(I)$ the better. 
\begin{figure}[h]
\begin{centering}
\includegraphics[scale=0.5]{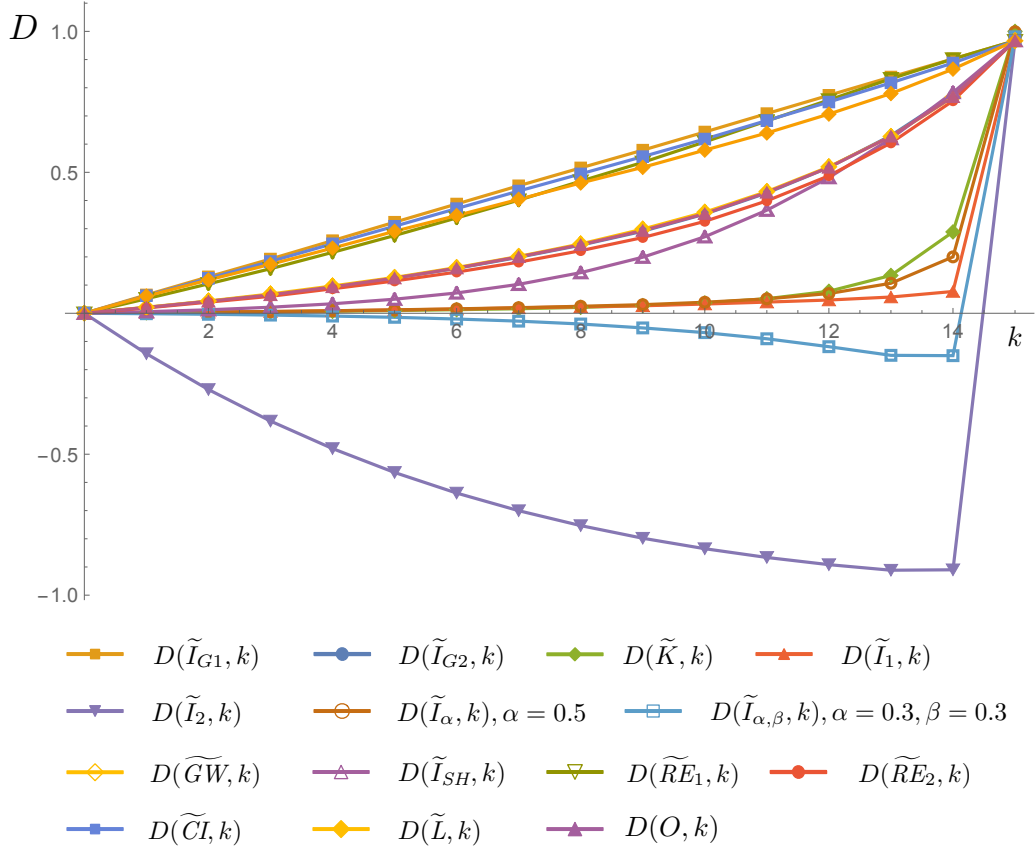}
\par\end{centering}
\caption{Rescaled ordered distance for different inconsistency indices for
incomplete PC matrices with $k$ missing comparisons.}

\label{fig:directed-distance-ind-chart1}
\end{figure}

In Figure \ref{fig:directed-distance-ind-chart1} there are fourteen
plots\footnote{The exact numerical data are presented in the Appendix in Table \ref{tab:exact-data-table}.}
corresponding to the average ordered distance (\ref{eq:average-ordered-distance})
for all the inconsistency indices introduced in \ref{sec:Matrix-based-indices}
and \ref{sec:Ranking-based-indices} and three additional indices
found in the literature. It is easy to see that, in general, the matrix
based indices perform better than the ranking based indices. The exception
here is the index $\widetilde{I}_{2}$, which very quickly reveals
high differences between complete and incomplete PC matrices. It is
interesting that the incomplete matrices are considered by this index
as more inconsistent than the complete ones (most of the plot is below
the abscissa). The behavior of $\widetilde{I}_{2}$ is inherited by
the $\widetilde{I}_{\alpha,\beta}$ index. Here, one can also notice
that incomplete matrices are considered as more inconsistent than
their complete counterparts. Fortunately, the other matrix based indices
perform very well. The best of them is $\widetilde{I}_{1}$, where
$\mathscr{D}(\widetilde{I}_{1})=1.4256$. Then $\mathscr{D}(\widetilde{I}_{\alpha})=1.5961$,
and next $\mathscr{D}(\widetilde{I}_{\alpha,\beta})=1.6882$. Among
the ranking based indices, the modification of Salo-Hamalainen index
$\widetilde{I}_{\textit{SH}}$ stands out positively as it gets $\mathscr{D}(\widetilde{I}_{\textit{SH}})=4.1286$.
The other ranking based indices, like modified Barzilai's relative
error index version 1, get higher areas under the plot. Thus, incompleteness
influences the assessment of the degree of inconsistency to a greater
extent than for previous indices. All the values of $\mathscr{D}$
are shown in Table \ref{tab:d-res}.

\begin{table}
\setlength\extrarowheight{3pt}
\setlength{\tabcolsep}{2pt}
\begin{centering}
\begin{tabular}{|c|c|c|c|}
\hline 
Pos. & Notation & Name & $\mathscr{D}(I)$\tabularnewline
\hline 
\hline 
$1$. & $\widetilde{I}_{1}$ & Cycle based index v. I & $1.3368$\tabularnewline
\hline 
$2$. & $\widetilde{I}_{\alpha}$ & $\alpha$-index, for $\alpha=0.5$ & $1.5896$\tabularnewline
\hline 
$3$. & $\widetilde{K}$ & Koczkodaj index & $1.692$\tabularnewline
\hline 
$4$. & $\widetilde{I}_{\alpha,\beta}$ & $\alpha,\beta$-index, for $\alpha=\beta=0.3$ & $1.7321$\tabularnewline
\hline 
$5$. & $\widetilde{I}_{\textit{SH}}$ & Salo-Hamalainen index & $4.1286$\tabularnewline
\hline 
$6$. & $\widetilde{\textit{RE}}_{2}$ & Barzilai's relative error index v. II & $4.6748$\tabularnewline
\hline 
$7$. & $\widetilde{I}_{G2}$ & Geometric consistency index v. II & $4.9048$\tabularnewline
\hline 
$8.$ & $O$ & Oliva-Setola-Scala's index & $4.9062$\tabularnewline
\hline 
$9$. & $\widetilde{GW}$ & Golden-Wang index & $4.9359$\tabularnewline
\hline 
$10.$ & $\widetilde{L}$ & Logarithmic least square condition & $7.1425$\tabularnewline
\hline 
$11$. & $\widetilde{\textit{RE}}_{1}$ & Barzilai's relative error index v. I & $7.2958$\tabularnewline
\hline 
$12$. & $\widetilde{\textit{CI}}$ & Saaty consistency index & $7.5028$\tabularnewline
\hline 
$13$. & $\widetilde{I}_{G1}$ & Geometric consistency index v. I & $7.7333$\tabularnewline
\hline 
$14$. & $\widetilde{I}_{2}$ & Cycle based index v. II & $10.112$\tabularnewline
\hline 
\end{tabular}
\par\end{centering}
\caption{The total distance $\mathscr{D}$ from the abscissae of the $D(I,k)$
plots for all considered indices. The smaller the value, the more
robust the given inconsistency index.}

\label{tab:d-res}
\end{table}

\section{Discussion and summary\label{sec:Discussion}\label{sec:Summary}}

The four best indices (Table \ref{tab:d-res}) achieve very similar
results. They are all very good, which means that the differences
in inconsistency measured by these indices between complete and incomplete
matrices is small. For instance, $D(\widetilde{I}_{1},4)=0.009$ and
$D(\widetilde{I}_{1},11)=0.045$, which means that for $20\%$ of
missing comparisons we may expect a difference in value of the index
smaller than $1\%$, and for $50\%$ of missing comparisons this difference
should not be greater than $4.5\%-5\%$. Such results clearly show
that inconsistency measurement for incomplete PC matrices can indeed
be a valuable indication of the quality of decision data. Thus, methods
for calculating the ranking for incomplete PC matrices mentioned in
Section \ref{subsec:Incompleteness} also gain methods for estimating
data inconsistency.

An obvious disadvantage of the matrix based indices defined in Section
\ref{sec:Matrix-based-indices} is the need to find all cycles in
the matrix graph. This can be particularly difficult and time-consuming
for larger matrices. A way to deal with a large number of cycles may
be to limit their number. This might be achieved by limiting the analysis
of inconsistency to fundamental cycles only, or just to a random set
of cycles.  A similar problem does not occur in the case of the ranking
based indices. The best of them, the modification of Salo-Hamalainen
index for incomplete PC matrices, gets the total distance $\text{\ensuremath{\mathscr{D}(\widetilde{I}_{\textit{SH}})=}}4.1286$.
It is also a good result, which proves that this index can be effectively
used to assess the inconsistency of incomplete PC matrices.   

The article presents extensions for twelve inconsistency indices that
allow them to also be used for incomplete PC matrices. Thanks to this,
users of the pairwise comparison method (including AHP) receive a
way to determine the quality of incomplete decision data . The presented
research does not determine which of the defined indices is the best
in practice. Robust indices can be difficult to implement and calculate.
On the other hand, indices that are easier to calculate can be more
vulnerable for decision data deterioration. Finding a solution that
combines robustness with the simplicity of implementation and calculation
will still be a challenge for researchers.

\section*{Acknowledgment}

The authors would like to show their gratitude to José María Moreno-Jiménez
(Universidad de Zaragoza, Spain), Sándor Bozóki (Hungarian Academy
of Sciences and Corvinus University of Budapest, Hungary) for their
comments on the early version of the paper. Special thanks are due to Ian Corkill for his editorial help.

\section*{Disclosure statement}

No potential conflict of interest was reported by the authors.

\section*{Funding}

The research is supported by The National Science Centre (Narodowe
Centrum Nauki), Poland, project no. 2017/25/B/HS4/01617.

\bibliographystyle{plain}
\bibliography{papers_biblio_reviewed}

\appendix

\section*{Appendix}

\begin{sidewaystable}
\setlength\extrarowheight{3pt}
\setlength{\tabcolsep}{1pt}%
\begin{tabular}{l|cccccccccccccc}
\rowcolor{gray1} $k$ & $D(\widetilde{I}_{G1},k)$ & $D(\widetilde{I}_{G2},k)$ & $D(\widetilde{K},k)$ & $D(\widetilde{I}_{1},k)$ & $D(\widetilde{I}_{2},k)$ & $D(\widetilde{I}_{\alpha},k)$ & $D(\widetilde{I}_{\alpha,\beta},k)$ & $D(\widetilde{GW},k)$ & $D(\widetilde{I}_{\textit{SH}},k)$ & $D(\widetilde{\textit{RE}}_{1},k)$ & $D(\widetilde{\textit{RE}}_{2},k)$ & $D(\widetilde{\textit{CI}},k)$ & $D(\widetilde{L},k)$ & $D(O,k)$\tabularnewline
\hline 
0 & 0.000 & 0.000 & 0.000 & 0.000 & 0.000 & 0.000 & 0.000 & 0.000 & 0.000 & 0.000 & 0.000 & 0.000 & 0.000 & 0.000\tabularnewline
\rowcolor{gray1}1 & 0.065 & 0.021 & 0.002 & 0.002 & -0.143 & 0.002 & -0.002 & 0.021 & 0.006 & 0.051 & 0.019 & 0.062 & 0.061 & 0.020\tabularnewline
2 & 0.129 & 0.043 & 0.003 & 0.004 & -0.271 & 0.004 & -0.004 & 0.043 & 0.012 & 0.104 & 0.039 & 0.124 & 0.118 & 0.042\tabularnewline
\rowcolor{gray1}3 & 0.192 & 0.066 & 0.005 & 0.006 & -0.382 & 0.006 & -0.006 & 0.068 & 0.021 & 0.158 & 0.060 & 0.184 & 0.174 & 0.065\tabularnewline
4 & 0.258 & 0.093 & 0.007 & 0.009 & -0.480 & 0.009 & -0.010 & 0.096 & 0.034 & 0.216 & 0.087 & 0.246 & 0.231 & 0.094\tabularnewline
\rowcolor{gray1}5 & 0.323 & 0.125 & 0.010 & 0.012 & -0.565 & 0.012 & -0.014 & 0.127 & 0.050 & 0.276 & 0.115 & 0.309 & 0.291 & 0.125\tabularnewline
6 & 0.388 & 0.159 & 0.013 & 0.015 & -0.638 & 0.016 & -0.020 & 0.162 & 0.072 & 0.338 & 0.146 & 0.371 & 0.347 & 0.160\tabularnewline
\rowcolor{gray1}7 & 0.452 & 0.199 & 0.016 & 0.019 & -0.700 & 0.020 & -0.028 & 0.201 & 0.103 & 0.401 & 0.181 & 0.434 & 0.404 & 0.200\tabularnewline
8 & 0.516 & 0.242 & 0.021 & 0.023 & -0.753 & 0.025 & -0.038 & 0.246 & 0.144 & 0.469 & 0.222 & 0.494 & 0.462 & 0.243\tabularnewline
\rowcolor{gray1}9 & 0.579 & 0.291 & 0.027 & 0.027 & -0.798 & 0.031 & -0.052 & 0.298 & 0.199 & 0.535 & 0.269 & 0.556 & 0.518 & 0.294\tabularnewline
10 & 0.643 & 0.352 & 0.037 & 0.033 & -0.835 & 0.039 & -0.069 & 0.358 & 0.271 & 0.608 & 0.326 & 0.618 & 0.579 & 0.353\tabularnewline
\rowcolor{gray1}11 & 0.708 & 0.427 & 0.052 & 0.040 & -0.866 & 0.051 & -0.090 & 0.431 & 0.366 & 0.682 & 0.399 & 0.684 & 0.640 & 0.429\tabularnewline
12 & 0.773 & 0.517 & 0.078 & 0.047 & -0.892 & 0.070 & -0.119 & 0.519 & 0.481 & 0.757 & 0.488 & 0.749 & 0.707 & 0.518\tabularnewline
\rowcolor{gray1}13 & 0.838 & 0.630 & 0.134 & 0.058 & -0.911 & 0.106 & -0.149 & 0.627 & 0.620 & 0.832 & 0.604 & 0.817 & 0.779 & 0.628\tabularnewline
14 & 0.902 & 0.774 & 0.289 & 0.077 & -0.910 & 0.200 & -0.150 & 0.766 & 0.783 & 0.902 & 0.754 & 0.888 & 0.866 & 0.768\tabularnewline
\rowcolor{gray1}15 & 0.967 & 0.967 & 1.000 & 0.967 & 0.967 & 0.999 & 0.981 & 0.966 & 0.967 & 0.967 & 0.967 & 0.967 & 0.967 & 0.967\tabularnewline
\multicolumn{1}{l}{} &  &  &  &  &  &  &  &  &  &  &  &  &  & \tabularnewline
\rowcolor{gray1} & $\mathscr{D}(\widetilde{I}_{G2})$ & $\mathscr{D}(\widetilde{I}_{G1})$ & $\mathscr{D}(\widetilde{K})$ & $\mathscr{D}(\widetilde{I}_{1})$ & $\mathscr{D}(\widetilde{I}_{2})$ & $\mathscr{D}(\widetilde{I}_{\alpha})$ & $\mathscr{D}(\widetilde{I}_{\alpha,\beta})$ & $\mathscr{D}(\widetilde{GW})$ & $\mathscr{D}(\widetilde{I}_{\textit{SH}})$ & $\mathscr{D}(\widetilde{\textit{RE}}_{1})$ & $\mathscr{D}(\widetilde{\textit{RE}}_{2})$ & $\mathscr{D}(\widetilde{\textit{CI}})$ & $\mathscr{D}(\widetilde{L})$ & $\mathscr{D}(O)$\tabularnewline
\hline 
 & 7.733 & 4.905 & 1.692 & 1.337 & 10.112 & 1.590 & 1.732 & 4.936 & 4.129 & 7.296 & 4.675 & 7.503 & 7.143 & 4.906\tabularnewline
\end{tabular}
\end{sidewaystable}

\end{document}